\documentclass{llncs}
\usepackage{makeidx}
\usepackage{graphicx}
 
\usepackage{ifthen}
\usepackage{microtype}
\usepackage{boxedminipage}
\usepackage{amsmath}
\usepackage{amsfonts}
\usepackage{amssymb}
\usepackage{authblk}

\usepackage{float}
\floatstyle{boxed}
\restylefloat{figure}

\newcommand{\pr}{{\bf Pr}}

\newcommand{\etal}{{\it et al.}}

\begin{document}
 
\title{A bad 2-dimensional instance for k-means++}
%
%
\author{Ragesh Jaiswal \and Prachi Jain \and Saumya Yadav}
\authorrunning{Jaiswal et al.} 
%
%
\institute{Department of Computer Science and Engineering, \\
Indian Institute of Technology Delhi, New Delhi, India.
}

\maketitle     

\begin{abstract}
The k-means++ seeding algorithm is one of the most popular algorithms that is used for finding the initial $k$ centers when using the k-means heuristic. The algorithm is a simple sampling procedure and can be described as follows: 
\begin{quote}
Pick the first center randomly from among the given points. 
For $i > 1$, pick a point to be the $i^{th}$ center with probability proportional to the square of the Euclidean distance of this point to the previously $(i-1)$ chosen centers.
\end{quote}
The k-means++ seeding algorithm is not only simple and fast but gives an $O(\log{k})$ approximation in expectation as shown by Arthur and Vassilvitskii~\cite{av07}.
There are datasets~\cite{av07,adk09} on which this seeding algorithm gives an approximation factor $\Omega(\log{k})$ in expectation. 
However, it is not clear from these results if the algorithm achieves good approximation factor with reasonably large probability (say $1/poly(k)$). 
Brunsch and R\"{o}glin~\cite{br11} gave a dataset where the k-means++ seeding algorithm achieves an approximation ratio of $(2/3 - \epsilon)\cdot \log{k}$ only with probability that is exponentially small in $k$. 
However, this and all other known {\em lower-bound examples} \cite{av07,adk09} are high dimensional. 
So, an open problem is to understand the behavior of the 
algorithm on low dimensional datasets. 
In this work, we give a simple two dimensional dataset on which the seeding algorithm achieves an approximation ratio $c$ (for some universal constant $c$) only with probability exponentially small in $k$. 
This is the first step towards solving open problems posed by Mahajan \etal~\cite{mnv12} and by Brunsch and R\"{o}glin~\cite{br11}.
\end{abstract}

\section{Introduction}

The k-means clustering problem is one of the most important problems in Data Mining and Machine Learning that has been widely studied. The problem is defined as follows:
\begin{quote}
{\bf (k-means problem)}: Given a set of $n$ points $X = \{x_1, ..., x_n\}$ in a $d$-dimensional space,
find a set of $k$ points $C = \{c_1, ..., c_k\}$ such that the cost function $\phi_{C}(X) = \sum_{x \in X} \min_{c \in C} D(x, c)$ is minimized. Here $D(x,c)$ denotes the square of the Euclidean distance between points $x$ and $c$.
In the {\em discrete} version of this problem the centers are constrained to be a subset of the given points $X$.
\end{quote}

The problem is known to be NP-hard even for small values of the parameters such as when $k=2$~\cite{d07} and when 
$d=2$~\cite{v09,mnv12}.
There are various approximation algorithm for solving the problem.
However, in practice, a heuristic known as the k-means algorithm (also known as Lloyd's algorithm) is used because of its excellent performance on real datasets even though it does not given any performance guarantees. 
This algorithm is simple and can be described as follows: 
\begin{quote}
{\bf (k-means Algorithm)}: (i) Arbitrarily, pick $k$ points $C$ as centers. (ii) Cluster the given points based on the nearest distance of points to centers in $C$. (iii) For all clusters, find the mean of all points within a cluster and replace the corresponding member of $C$ with this mean. Repeat steps (ii) and (iii) until convergence.
\end{quote}

Even though the above algorithm performs very well on real datasets, it does not have any performance guarantees. 
This means that this {\em local search} algorithm may either converge to a local optimum solution or may take a large amount of time to converge~\cite{av06a,av06b}. Poor choice of the initial $k$ centers (step (i)) is one of the main reasons for its bad performance with respect to approximation factor. 
A number of {\em seeding} heuristics have been suggested for picking the initial centers. 
One such seeding algorithm that has become popular is the k-means++ seeding algorithm.
The algorithm is extremely simple and runs very fast in practice.
Moreover, this simple randomized algorithm also gives an approximation factor of $O(\log{k})$ in expectation~\cite{av07}.
In practice, this seeding technique is used for find the initial $k$ centers to be used with the k-means algorithm and this guarantees a theoretical approximation guarantee.
The simplicity of the algorithm can be seen by its simple description below:
\begin{quote}
{\bf (k-means++ seeding)}: Pick the first center randomly from among the given points. Pick a point to be the $i^{th}$ center ($i>1$) with probability proportional to the square of the Euclidean distance of this point to the previously $(i-1)$ chosen centers.
\end{quote}

A lot of recent work has been done in understanding the power of this simple sampling based approach for clustering.
We discuss these in the following paragraph.

\subsection{Related work} 
Arthur and Vassilvitskii~\cite{av07} showed that the sampling algorithm gives an approximation guarantee of $O(\log{k})$ in expectation. They also give an example dataset on which this approximation guarantee is best possible. Ailon \etal~\cite{ajm09} and Aggarwal \etal~\cite{adk09} show that sampling more than $k$ centers in the manner described above gives a constant pseudo-approximation. 
Ackermann and Bl\"{o}mer~\cite{ab10} showed that the results of Arthur and Vassilvitskii~\cite{av07} may be extended to a large class of other distance measures. Jaiswal \etal~\cite{jks12} showed that the seeding algorithm may be modified appropriately to give a $(1 + \epsilon)$-approximation algorithm for the k-means problem. 
Jaiswal and Garg~\cite{jg12} and Agarwal \etal~\cite{ajp13} showed that if the dataset satisfies certain separation conditions, then the seeding algorithm gives constant approximation with probability $\Omega(1/k)$. 
Bahmani \etal~\cite{b12} showed that the seeding algorithm performs well even when fewer than $k$ sampling iterations are executed provided that more than one center is chosen in a sampling iteration.
We now discuss our main results.

\subsection{Main Results}
The lower-bound examples of Arthur and Vassilvitskii~\cite{av07} and Aggarwal \etal~\cite{adk09} have the following two properties: (a) the examples are high dimensional and (b) the examples lower-bound the {\em expected} approximation factor. 
Brunsch and R\"{o}glin~\cite{br11} discussed whether the k-means++ seeding gives better than $O(\log{k})$ approximation with probability $\Omega(1/poly(k))$. 
They constructed a high dimensional example where this is not true and a $O(\log{k})$ approximation is achieved only with probability exponentially small in $k$.
An important open problem mentioned in their work is to understand the behavior of the seeding algorithm on low-dimensional examples. 
This problem is also mentioned as an open problem by Mahajan \etal~\cite{mnv12} who examined the hardness of the k-means problem on $2$-dimensional datasets. 
In this work, we construct a two dimensional dataset on which the k-means++ seeding algorithm achieves an approximation ratio of $c$ (for some universal constant $c$) with probability exponentially small in $k$. 
Following is the main theorem that we prove in this work.

\begin{theorem}[Main Theorem]
Let $k\geq 10^3$ and $\eta = 0.999$. Consider the discrete version of the k-means problem.
There exists a two dimensional dataset $X$ such that the probability that the k-means++ algorithm 
gives an approximation factor better than $(9-\eta)/8$ on $X$ is at most $(2\sqrt{k}) \cdot 2^{-k/300}$.
\end{theorem}

For the non-discrete version, we get the above statement with approximation factor $\left(\frac{9-\eta}{4} \right)$.

\subsection{Our techniques}
All the known lower-bound examples~\cite{av07,adk09,br11} have the following general properties:
\begin{enumerate}
\item[(a)] All optimal clusters have equal number of points.
\item[(b)] The optimal clusters are high dimensional simplices.
\end{enumerate}
In order to construct a counterexample for the two dimensional case, we consider datasets that have different number of points in different optimal clusters. 
Our counterexample is shown in Figure~\ref{fig:example}.
Note that the optimal clusters are points at the end of the vertical bars
and the cluster sizes decreases exponentially going from left to right.
We say that the seeding algorithm {\em covers} the $i^{th}$ optimal clusters if the algorithm picks a point from either of the ends of the $i^{th}$ vertical bar (recall that this is the $i^{th}$ optimal cluster). 
The proof follows from the following three observations about this dataset:
\begin{itemize}
\item {\bf Observation 1}: Once the $i^{th}$ cluster gets covered, the probability of covering any cluster $j$ in subsequent rounds is roughly the same for any $j > i$. Moreover, there is a good chance that after the first few iterations, a cluster $i$ for some small $i$ gets covered.

\item {\bf Observation 2}: The algorithm needs to cover more than some constant fraction of clusters to achieve good (another constant) approximation.

\item {\bf Observation 3}: Given that a small numbered cluster (clusters are numbered from left to right) is covered in the initial few iterations, the probability of covering more than certain constant fraction of clusters is exponentially small in $k$.
\end{itemize}

In the next section, we give the details of this proof.

\section{Proof of the Main Theorem}\label{sec:2}

\begin{figure}[h]
  \centering
    \includegraphics[scale=0.5]{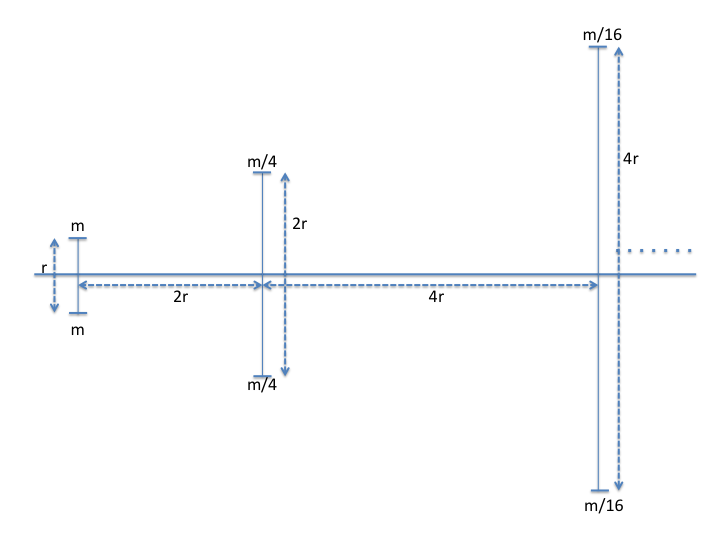}
    \vspace{-10pt}
    \caption{The 2-dimensional instance on which the k-means++ algorithms behaves badly.}
    \label{fig:example}
\end{figure}

The figure means that there are $m$ points each located at $(0, r/2)$ and $(0, - r/2)$. 
There are $m/4$ points each located at $(2r, r)$ and $(2r, -r)$ and so on. 
Here are some simple observations regarding this example:
The total number of points is $2 \cdot \sum_{i=1}^{k} \frac{m}{2^{2(i-1)}}$.
Note that the optimal cost in the discrete version of the problem is $k m r^2$.
Moreover, this is when one center from each of the $k$ clusters (vertical bars) is chosen.
We denote the optimal clusters by $A_1, ..., A_k$ from left to right. 
Note that the number of points in these clusters drops exponentially. 
We say that an optimal cluster $A_j$ is {\em covered} by the k-means++ algorithm if the algorithm picks a point as
a center from $A_j$.

Since the number of points in the initial few clusters are large, there is a good chance 
that in the initial few iterations of the k-means++ seeding algorithm, one center from 
these initial few clusters are chosen. 
The next lemma shows this more formally. 

\begin{lemma}\label{lemma-1}
Let $0 < \alpha, \beta \leq 1$and let $L = \{A_1, ..., A_{\beta k}\}$. Let $C = \{c_1, ...,c_{\alpha k}\}$ be the set of centers chosen by the k-means++ algorithm in the first $\alpha k$ iterations. Then
\[
\pr[C \cap (A_1 \cup ... \cup A_{\beta k}) = \phi] \leq e^{-(\alpha \beta/3) k}
\]
\end{lemma}
\begin{proof}
For $1 \leq i \leq \alpha k$, let $E_i$ denote the event that $c_i$ does not cover any cluster in $L$. We make the following observations: 
\[
\pr[E_1] = 1 - \frac{m + m/2^2 + ... + m/2^{2 \beta k - 2}}{m + m/2^2 +  ... + m/2^{2k-2}} = 1 - \frac{1 - 1/2^{2\beta k}}{1 - 1/2^{2k}} < 1/2^{2\beta k}
\]
We can also prove the following simple lemma.
\begin{lemma}
$\forall i > 1, \pr[E_i | E_1,...,E_{i-1}] \leq (1 - \beta/3)$
\end{lemma}
\begin{proof}
Consider centers $C_{i-1} = \{c_1, c_2, ..., c_{i-1}\}$. 
Let $j$ be the smallest integer such that $C_{i-1}$ covers $A_j$. Conditioned on the event $E_1, ..., E_{i-1}$, we have that $j > \beta k$. Let us partition the optimal clusters into the following 3 parts: the first partition is $L$, the second partition is $M = \{A_{\beta k + 1}, ..., A_{j-1}\}$ and the third partition is $R = \{A_{j}, ..., A_{k}\}$. We note that
\[
\pr[\neg E_i | E_1, ..., E_{i-1}] = \frac{\phi_{C_{i-1}}(L)}{\phi_{C_{i-1}}(L) + \phi_{C_{i-1}}(M) + \phi_{C_{i-1}}(R)}
\]
Note that $\phi_{C_{i-1}}(R) \leq 5 (k - \beta k - i)mr^2 \leq 5 kmr^2$, $\phi_{C_{i-1}}(M) \leq \phi_{C_{i-1}}(L)$, and $\phi_{C_{i-1}}(L) \geq 10 \cdot 2^{2 \beta k} mr^2$.
Using these inequalities we get the following:
\begin{eqnarray*}
&& \pr[\neg E_i | E_1, ..., E_{i-1}] \geq \frac{1}{2 + (1/2)\frac{k}{2^{2\beta k}}} \geq \frac{\beta}{3}\\
\Rightarrow && \pr[E_i | E_1, ..., E_{i-1}] \leq 1-\frac{\beta}{3}
\end{eqnarray*} \qed
\end{proof}
So we get that $\pr[E_1, ..., E_{\alpha k}] \leq (1 - \beta/3)^{\alpha k} \leq e^{-(\alpha \beta/3) k}$. This completes the proof of Lemma~\ref{lemma-1}. \qed
\end{proof}

The next lemma shows that unless a large number of optimal clusters get covered, the approximation factor is bad.

\begin{lemma}\label{lemma-2}
Let $C$ denote the centers chosen by the k-means++ algorithm. If $C$ covers $\leq \alpha \cdot k$ clusters, then
\[
\frac{\phi_{C}(X)}{\phi_{OPT}(X)} \geq \frac{9 - \alpha}{8}
\]
\end{lemma}
\begin{proof}
Note that $\phi_{OPT}(X) = kmr^2$. 
This is when all the optimal clusters are covered. 
Suppose $C$ is such that $\leq \alpha k$ clusters are covered by $C$. 
Let $A_i$ be an optimal cluster that is not covered with respect to $C$.
Then we have:
\begin{eqnarray*}
\phi_C(A) &\geq& 2 \cdot \frac{m}{2^{2i-2}} \cdot \left( (2^{i-2} r - 2^{i-3} r)^2 + (2^{i-1} r)^2\right) \\
&=& 2m \cdot ( r^2/16 + r^2) \\
&=& mr^2 \cdot (2 + 1/8) = (17/8) mr^2
\end{eqnarray*}
Using this, we have:
\begin{eqnarray*}
\phi_{C}(X) &\geq& mr^2 \cdot (\alpha k - (1-\alpha)k) + (1 - \alpha) \cdot k \cdot (17/8)mr^2 \\
&=& kmr^2(2\alpha - 1) + (17/8)(1-\alpha)k m r^2 \\
&=& kmr^2 \cdot \frac{9 -  \alpha}{8}
\end{eqnarray*}
So, we get $\frac{\phi_{C}(X)}{\phi_{OPT}(X)} \geq \frac{9-\alpha}{8}$. \qed

\end{proof}

We now need to show that the probability that k-means++ algorithm covers more than $\eta k$ (for some constant $\eta$) clusters is exponentially 
small in $k$. 
To prove this, we define and analyze a random sampling procedure which may be of independent interest.

\subsection{A Sampling Problem}
In this section, we analyze a sampling procedure that will help in the analysis of the behavior of k-means++ for our counterexample. 
This might be of independent interest.
The procedure is defined as follows:
\begin{quote}
{\bf SampBall}: There are $2k$ balls each colored with one of $k$ colors such that for each color there are exactly two balls with that color. 
$k$ balls are sampled randomly without replacement out of these $2k$ balls.
\end{quote}

Let $B$ be the random variable denoting the sampled set of of size $k$. We are interested in bounding the probability that $B$ contains balls with more than $7k/8$ different colors. 
Let $E_i$ be the probability that $B$ contains of exactly $i$ distinct colors.
The next lemma bounds the probability of the event $E_i$.

\begin{lemma}
For any $i \geq 7k/8$, $\pr[E_i] \leq \frac{5}{\sqrt{k}} \cdot 2^{-k/16}$.
\end{lemma}
\begin{proof}
Let us consider the following alternative procedure $\mathcal{P}'$: whenever a new colored ball is sampled, the procedure outputs the letter ``N'' and if a ball with the same color has already been sampled, then $\mathcal{P}'$ outputs ``O''. 
Note that $\pr[E_i]$ is equal to the probability that $\mathcal{P}'$ outputs a string in $\{N, O\}^k$ such that there are exactly $i$ N's. Let $S$ be the random variable denoting the string output of $\mathcal{P}'$. Then we have:
\begin{eqnarray*}
\pr[E_i] &=& \pr[S \textrm{ has exactly } i \textrm{ N's}] \\
&\leq& \binom{k}{i} \cdot \pr[S = \underbrace{NN...N}_{i\ terms} \underbrace{OO...O}_{(k-i)\ terms}] \\
&=& \binom{k}{i} \cdot \frac{2k}{2k} \cdot \frac{2(k-1)}{2(k-1) + 1} ... \frac{2(k-i+1)}{2(k-i+1) + i-1} \cdot \frac{i}{i+2(k-i)} ...\frac{2i-k+1}{k+1}\\
&=& \binom{k}{i} \cdot \frac{2^i \cdot (k!)^2 \cdot i!}{(k-i)! \cdot (2k)! \cdot (2i-k)!}\\
&=& \frac{2^i \cdot (k!)^3}{((k-i)!)^2 \cdot (2k)! \cdot (2i-k)!} \\
&\leq& \frac{2^{7k/8} \cdot (k/e)^{3k} \cdot e^3 \cdot k^{3/2}}{2\pi \cdot (k/8) \cdot (k/8e)^{k/4} \cdot \sqrt{2\pi} \cdot \sqrt{2k} \cdot (2k/e)^{2k} \cdot \sqrt{2\pi} \cdot \sqrt{k/2} \cdot (3k/4e)^{3k/4}} \\
&=& \frac{2 \cdot e^3 \cdot 2^{3k/4 + 7k/8} \cdot (4/3)^{3k/4}}{\pi^2 \cdot \sqrt{k} \cdot 2^{2k}} \\
&\leq& \frac{2 \cdot e^3 \cdot 2^{3k/4 + 7k/8 + 5k/16}}{\pi^2 \cdot \sqrt{k} \cdot 2^{2k}} \\
&\leq& \frac{5}{\sqrt{k}} \cdot 2^{-k/16}
\end{eqnarray*}\qed
\end{proof}

This gives us the following useful corollary.

\begin{corollary}
$\pr[B \textrm{ has more than $7k/8$ colored balls}] \leq 5 \sqrt{k} \cdot 2^{-k/16}$.
\end{corollary}

The relationship of the sampling procedure {\bf SampBall} with the counter-example should not be very difficult to see. Sampling a ball with new color corresponds to sampling a center from an uncovered cluster and so on. 
The main difference is that sampling a center from a new cluster is more likely than sampling a ball with new color. 
We modify our sampling procedure to be able to use the analysis for analyzing k-means++ over our counterexample. 
Here is our new sampling procedure:

\begin{quote}
{\bf BiasedSampBall}: There are $2k$ balls each colored with one of $k$ colors such that for each color there are two balls with that color. $k$ balls are sampled randomly without replacement out of these $2k$ balls. There is a bias towards sampling balls of new color. When sampling a ball the probability of sampling a ball of new color is at most $\gamma$ times more than the probability of sampling a ball with color that has already been picked.
\end{quote}
Any value of $1 \leq \gamma \leq 5$ will work for our purposes. We are interested in the probability that the above randomized procedure picks balls of at least $(0.99)k$ different colors. Let $B$ be the random variable denoting the sample of $k$ balls.
Let $E_i$ be the probability that $B$ contains of exactly $i$ distinct colors.
Next, we bound the probability of the event $E_i$.
First, we need the following simple lemma bounding a quantity we will later need.

\begin{lemma}\label{lemma:calc}
$\prod_{j=1}^{i} (2k - (j-1)\cdot \frac{9}{5}) \geq 2^i \cdot \frac{k!}{\left( k - \frac{9}{10}\cdot i \right)!} \cdot (k-\frac{9}{10}\cdot i)^{i/10}$.
\end{lemma}
\begin{proof}
We show the above lemma using the following calculations:
\begin{eqnarray*}
\prod_{j=1}^{i} (2k - (j-1)\cdot \frac{9}{5}) &=& \left(2k \right) \cdot \left(2k - \frac{9}{5} \right) \cdot \left(2k - \frac{18}{5}\right) \cdot \left(2k - \frac{27}{5} \right) \ldots \left(2k-\frac{9}{5}(i-1)\right) \\
&\geq& (2k) \cdot (2k - 2) \cdot (2k - 4) \cdot (2k - 6)...(2k - 8) \cdot  \\
&& (2k - 9)\cdot (2k - 11)\cdot (2k - 13) \cdot (2k - 15) \cdot (2k - 17) \cdot \\
&& (2k - 18) \cdot (2k - 20)\cdot (2k - 22) \cdot (2k - 24) \cdot (2k - 26) \cdot \\
&& ...\left(2k-\frac{9}{5}(i-1)\right) \\
&\geq& 2^i \cdot (k) \cdot (k - 1) \cdot (k-2) \cdot (k-3) \cdot (k-4) \cdot \\
&& (k-5) \cdot (k-6) \cdot (k-7) \cdot (k-8) \cdot (k-9) \cdot \\
&& (k-9) \cdot (k-10) \cdot (k-11) \cdot (k-12) \cdot (k-13) \cdot \\
&& ...\left(k-\frac{9}{10}(i-1)\right) \\
&\geq& 2^i \cdot \frac{k!}{\left(k-\frac{9}{10}i\right)!} \cdot \left(k - \frac{9}{10}\cdot i \right)^{i/10}
\end{eqnarray*}\qed
\end{proof}

\begin{lemma}
For any $i \geq (0.99) k$ and $1 \leq \gamma \leq 5$, $\pr[E_i] \leq \frac{1}{\sqrt{k}} \cdot 2^{-k/64}$.
\end{lemma}
\begin{proof}
Let us consider the following alternative procedure $\mathcal{P}'$: whenever a new colored ball is sampled, the procedure outputs the letter ``N'' and if a ball with the same color has already been sampled, then $\mathcal{P}'$ outputs ``O''. 
Note that $\pr[E_i]$ is equal to the probability that $\mathcal{P}'$ outputs a string in $\{N, O\}^k$ such that there are exactly $i$ N's. Let $S$ be the random variable denoting the string output of $\mathcal{P}'$. Then we have:
\begin{eqnarray*}
\pr[E_i] &=& \pr[S \textrm{ has exactly } i \textrm{ N's}] \\
&\leq& \binom{k}{i} \cdot \pr[S = \underbrace{NN...N}_{i\ terms} \underbrace{OO...O}_{(k-i)\ terms}] \\
&=& \binom{k}{i} \cdot \frac{2\gamma k}{2\gamma k} \cdot \frac{2\gamma(k-1)}{2\gamma(k-1) + 1} ... \frac{2\gamma(k-i+1)}{2\gamma(k-i+1) + i-1} \cdot \frac{i}{i+2\gamma(k-i)} ...\frac{2i-k+1}{2i-k+1 + 2\gamma(k-i)}\\
&\leq & \binom{k}{i} \cdot \frac{2k}{2k} \cdot \frac{2(k-1)}{(2k-1) - 1 \cdot (1 - 1/\gamma)} ... \frac{2(k-i+1)}{(2k-i+1) - (i-1)\cdot (1-1/\gamma)} \cdot \\
&& \qquad \frac{i}{i+2\gamma(k-i)} ...\frac{2i-k+1}{2i-k+1 + 2\gamma(k-i)}\\
&\leq & \binom{k}{i} \cdot \frac{2k}{2k} \cdot \frac{2(k-1)}{(2k-1) - 1 \cdot (4/5)} \cdot \frac{2(k-2)}{(2k-2) - 2 \cdot (4/5)}... \frac{2(k-i+1)}{(2k-i+1) - (i-1)\cdot (4/5)} \cdot \\
&& \qquad \frac{i}{i+2(k-i)} ...\frac{2i-k+1}{2i-k+1 + 2(k-i)} \quad \textrm{(since $1 \leq \gamma \leq 5$)}\\
&\leq& \binom{k}{i} \cdot \frac{2^i \cdot k\cdot (k-1) ... (k-i+1)}{(2k)\cdot (2k - 1\cdot (9/5)) \cdot (2k - 2\cdot (9/5)) ... (2k - (i-1)\cdot (9/5))} \cdot
\frac{i \cdot (i-1)...(2i-k+1)}{(2k-i) \cdot (2k - i-1)...(k+1)}  \\
&=& \binom{k}{i} \cdot  \frac{k!}{(k-i)!} \cdot \frac{(k - (9/10)i)!}{k!} \cdot \frac{1}{(k - (9/10)i)^{i/10}} \cdot \frac{i!}{(2i-k)!} \cdot \frac{k!}{(2k-i)!} \textrm{(using Lemma~\ref{lemma:calc})}\\
&=& \binom{k}{i} \cdot  \frac{1}{(k-i)!} \cdot \frac{(k - (9/10)i)!}{1} \cdot \frac{1}{(k - (9/10)i)^{i/10}} \cdot \frac{i!}{(2i-k)!} \cdot \frac{k!}{(2k-i)!} \\
&=& \frac{k! \cdot (k - (9/10)i)!}{((k-i)!)^2} \cdot \frac{1}{(k - (9/10)i)^{i/10}} \cdot \frac{1}{(2i-k)!} \cdot \frac{k!}{(2k-i)!} \\
&\leq& \frac{1}{\sqrt{k}} \cdot 2^{-k/64}
\end{eqnarray*}
Note that the last step is obtained by using Sterling's approximation and plotting the resulting function. \qed
\end{proof}

This gives us the following useful corollary.

\begin{corollary}
$\pr[B \textrm{ has more than $(0.99) k$ colored balls}] \leq \sqrt{k} \cdot 2^{-k/64}$.
\end{corollary}

\subsection{k-means++ covers bounded fraction of clusters}

In this section, we will prove that the k-means++ algorithm covers at most $\eta k$ clusters for some universal constant $\eta$. 
This in conjunction with Lemma~\ref{lemma-2} gives the main theorem. 

\begin{lemma}\label{lemma:cover}
Let $k \geq 10^3$ and $\eta \geq (0.999)$. Then we have:
\[
\pr[\textrm{k-means++ covers more than } \eta k \textrm{ clusters}] \leq (2\sqrt{k}) \cdot 2^{-k/300}.
\]
\end{lemma}
\begin{proof}
Let $E$ denote the event that after the first $k/10$ iterations, the first $k/10$ cluster are uncovered. 
Note that from Lemma~\ref{lemma-1}, we have:
\[
\pr[E] \leq e^{-k/300}
\]
Note that conditioned on the event $\neg E$, we have that in the remaining $k(1 - 1/10)$ iterations the probability of sampling a center from either end of any uncovered cluster from the $k(1 - 1/10)$ rightmost clusters is at most $5$ times the probability of sampling from the uncovered end of a covered cluster. This is because the potential of the points at the uncovered end of a covered cluster is $mr^2$ and the that of points in either end of an uncovered cluster is at most $5mr^2$ given that event $\neg E$ happens.
\footnote{You may see this by computing the potential of points at one end of the rightmost cluster with respect to the center $(0, r/2)$}

So,we can use the analysis of the previous section. 
We get that conditioned on the event $\neg E$, the probability that more than $(0.99)k(1 - 1/10)$ from the $k(1-1/10)$ rightmost clusters will be covered is at most $(\sqrt{k}) \cdot 2^{-(0.9) k/64}$. So we have:
\begin{eqnarray*}
\pr[\textrm{k-means++ covers more than } \eta k \textrm{ clusters}] &\leq& 
\pr[E] +\\
&& \pr[\textrm{k-means++ covers more than } \eta k \textrm{ clusters} | \neg E] \\
&\leq& e^{-k/300} + (\sqrt{k}) \cdot 2^{-k/128}\\
&\leq& (2\sqrt{k}) \cdot 2^{-k/300}
\end{eqnarray*}\qed
\end{proof}

Now the proof of the main theorem follows from the above lemma and Lemma~\ref{lemma-2}.

\begin{proof}[Proof of main theorem]
From Lemmas~\ref{lemma-2} and \ref{lemma:cover}, we get that:
\[
\pr[\textrm{k-means++ gives better than } (9-\eta)/8 \textrm{ approximation}] \leq (2\sqrt{k})\cdot 2^{-k/300}.
\]
\qed
\end{proof}

\section{The k-median Problem}
The k-median problem is similar to the k-means problem.
Here, the objective is to minimize the sum of Euclidean distances rather 
than the sum of squares of the Euclidean distances as in the k-means problem. That is, the objective function to be minimized is $\phi_C(X) = \sum_{x \in X} \min_{c \in C} D(x, c)$, where $D(x, c)$ denotes the Euclidean distance between points $x$ and $c$. 
For the k-median problem, we make appropriate changes to the seeding algorithm. More specifically, we consider the following algorithm:
\begin{quote}
{\bf (SmpAlg)}: Pick the first center randomly from among the given points. Pick a point to be the $i^{th}$ center ($i>1$) with probability proportional to the Euclidean distance of this point to the previously $(i-1)$ chosen centers.
\end{quote}
The counterexample showing that the above algorithm achieves a fixed constant approximation with probability only exponentially small in $k$ is similar to the example in figure~1. Instead of the $i^{th}$ optimal cluster ($i^{th}$ vertical bar) containing $m/2^{2i-2}$ points at either end, it contains $m/2^{i-1}$ points at either end. 
Given this, note that the cost of the optimal clustering is $kmr$. 
Note that the analysis of the previous section can be easily extended for this case.

\section{Conclusions and Open Problems}
In this work, we give a two dimensional example dataset on which the k-means++ seeding algorithm achieves a constant factor approximation (for some universal constant) with probability exponentially small in $k$.
This is only the first step towards understanding the behavior of k-means++ seeding algorithm on low-dimensional datasets. 
This addresses the open question of Brunsch and R\"{o}glin~\cite{br11}. 
Brunsch and R\"{o}glin gave a $O(k^2)$-dimensional instance where the k-means++ seeding algorithm achieves $O(\log{k})$ approximation factor with exponentially small probability and ask whether similar instances can be constructed in small dimension. 
An interesting open question is whether we can show that the seeding algorithm gives better than $O(\log{k})$ approximation factor on instances in small dimension.

\section{Acknowledgements}
Ragesh Jaiswal would like to thank Nitin Garg and Abhishek Gupta who were involved in the initial stages of this project.




\end{document}